\documentclass[11pt, a4paper]{article}

\usepackage[authoryear]{natbib}
\usepackage{authblk}
\usepackage[dvipsnames]{xcolor}

\usepackage{lineno}
\nolinenumbers

\usepackage[breaklinks]{hyperref}
\hypersetup{
	colorlinks   = true, 
	urlcolor     = blue, 
	linkcolor    = MidnightBlue , 
	citecolor   = RedViolet 
}

\usepackage{amsmath,amsfonts,amsthm,amssymb}
\usepackage[retainorgcmds]{IEEEtrantools}
\usepackage{xspace}
\usepackage{url}
\usepackage{tikz}
\usepackage{enumitem}
\usepackage{booktabs}
\usepackage{microtype}

\usepackage[ruled,vlined,linesnumbered]{algorithm2e}
\providecommand{\DontPrintSemicolon}{\dontprintsemicolon}

\usepackage{soul}
\setuldepth{Berlin}

\theoremstyle{plain}
\newtheorem{theorem}{Theorem}[section]
\newtheorem*{mainthm}{Main Theorem}

\newtheorem{lemma}[theorem]{Lemma}
\newtheorem{proposition}[theorem]{Proposition}

\newtheorem{definition}[theorem]{Definition}
\newtheorem{corollary}[theorem]{Corollary}

\theoremstyle{definition}

\theoremstyle{remark}
\newtheorem{remark}[theorem]{Remark}

\newcommand{\bmu}{\ensuremath{\boldsymbol{\mu}}\xspace}
\newcommand{\mms}{\ensuremath{\textsc{mms}}\xspace}
\newcommand{\pmms}{\ensuremath{\textsc{pmms}}\xspace}
\newcommand{\gmms}{\ensuremath{\textsc{gmms}}\xspace}
\newcommand{\efo}{\ensuremath{\textsc{ef}\oldstylenums{1}}\xspace}
\newcommand{\efx}{\ensuremath{\textsc{efx}}\xspace}
\newcommand{\ef}{\ensuremath{\textsc{ef}}\xspace}

\newcommand{\mysetminusD}{\hbox{\tikz{\draw[line width=0.6pt,line cap=round] (3pt,0) -- (0,6pt);}}}
\newcommand{\mysetminusT}{\mysetminusD}
\newcommand{\mysetminusS}{\hbox{\tikz{\draw[line width=0.45pt,line cap=round] (2pt,0) -- (0,4pt);}}}
\newcommand{\mysetminusSS}{\hbox{\tikz{\draw[line width=0.4pt,line cap=round] (1.5pt,0) -- (0,3pt);}}}
\newcommand{\mysetminus}{\mathbin{\mathchoice{\mysetminusD}{\mysetminusT}{\mysetminusS}{\mysetminusSS}}}

\DeclareMathOperator*{\argmax}{arg\,max}

\allowdisplaybreaks

\usepackage[largesmallcaps]{kpfonts}
\usepackage[margin=1in]{geometry}

\title{Multiple Birds with One Stone: Beating $1/2$ for\\  EFX and GMMS via Envy Cycle Elimination}

\author[1,2]{Georgios Amanatidis}
\author[3]{Evangelos Markakis}
\author[3]{Apostolos Ntokos}
\affil[1]{University of Essex, {\small Department of Mathematical Sciences}}
\affil[2]{University of Amsterdam, {\small Institute for Logic, Language and Computation}}
\affil[3]{Athens University of Economics and Business, {\small Department of Informatics}}

\begin{document}
\maketitle

\begin{abstract}
Several relaxations of envy-freeness, tailored to fair division in settings with indivisible goods, have been introduced within the last decade. Due to the lack of general existence results for most of these concepts, great attention has been paid to establishing approximation guarantees. In this work, we propose a simple algorithm that is \emph{universally fair} in the sense that it returns allocations that have good approximation guarantees with respect to four such fairness notions at once. In particular, this is the first algorithm achieving a $(\phi -1)$-approximation of \emph{envy-freeness up to any good} (\efx) and a $\frac{2}{\phi +2}$-approximation of \emph{groupwise maximin share fairness} (\gmms), where $\phi$ is the golden ratio ($\phi \approx 1.618$). The best known approximation factor,  in polynomial time, for either one of these fairness notions prior to this work was $1/2$. Moreover, the returned allocation achieves \emph{envy-freeness up to one good} (\efo) and a $2/3$-approximation of \emph{pairwise maximin share fairness} (\pmms). 
While \efx is our primary focus, we also exhibit how to fine-tune our algorithm and further improve  the guarantees for \gmms or \pmms.

Finally, we show that $\gmms$---and thus $\pmms$ and $\efx$---allocations always exist when the number of goods does not exceed the number of agents by more than two. 
\end{abstract}


\section{Introduction}
\label{sec:intro}
The mathematical study of \emph{fair division} has a long and intriguing history, starting with the formal introduction of the cake-cutting problem by Banach, Knaster and Steinhaus \citep{Steinhaus48}. Ever since, we have seen the emergence of several fairness criteria, such as the classic notion of \emph{envy-freeness}, that has  a  dominant role in the literature, see e.g., \citet{BT96,Moulin03,RW98,HandbookComSoC2016}  and references therein. 

On the other hand, the computational study of finding fair allocations  when the resources are \emph{indivisible items} is more recent.  It is motivated by the realization that envy-freeness and other classic fairness notions 
are too demanding for the discrete setting. In particular, even with two agents and one item,
it is impossible to produce an allocation with \emph{any} reasonable worst-case approximation guarantee with respect to envy-freeness.

Within the last decade, these considerations have led to natural relaxations of envy-freeness, which are more suitable for the context of indivisible goods. The most prominent examples, that are also the focus of our work, include the notions of envy-freeness up to one good (\efo) and up to any good (\efx), maximin share fairness (\mms), as well as pairwise and groupwise maximin share fairness (\pmms and \gmms respectively). 
These relatively new concepts breathed new life into the field of fair division, but  they do not come without their issues. Most importantly, although they are generally easier to satisfy than envy-freeness, proving existence results has turned out to be a very challenging task---with the exception of \efo. For instance, it is an open problem to resolve whether \pmms and \efx allocations always exist, even for three and four agents with additive valuation functions, respectively. Surprisingly, existence remains unresolved even when the number of items 
is just slightly larger than the number of agents. 

A reasonable approach is to  focus on approximate versions of these relaxations. Indeed, this has led to a series of positive results, obtaining constant factor approximation algorithms for all the aforementioned relaxed criteria (see Related Work). 
However, improving on the currently known factors seems to be approaching a stagnation point. 
For example, soon after the  introduction of \efx, a $1/2$-approximation was established \citep{PR18}, but there has been no progress beyond $1/2$, despite the active interest on this notion.

We should also stress that these notions capture quite different aspects of fairness. A good approximation of any one of \efo, \efx, \mms and \pmms does not necessarily imply particularly strong guarantees for any of the others \citep{AmanatidisBM18}. Hence, it becomes compelling to ask for allocations that attain good guarantees with respect to several fairness notions simultaneously. Such results are rather scarce in the literature, e.g.,  \citep{BBMN18,GargM19} or are purely existential \citep{CaragiannisKMPS19}.

Motivated by the lack of such \emph{universally fair} algorithms, we look at the problem of computing allocations that (approximately) satisfy several fairness notions at the same time. Along the way, we aim to improve the  state-of-the-art for two of these notions, namely \efx and \gmms.
Somewhat unexpectedly, to do so we rely on simple subroutines that have been repeatedly used   in  fair division before. 

\medskip

\noindent\textbf{Contribution.} Our main contribution is an algorithm that is universally fair, in the sense that it achieves a better than $1/2$-approximation for all the notions under consideration. The main results can be summarized in the following statement.
\begin{mainthm}\label{thm:thm_main}
	We can efficiently compute an allocation that is simultaneously
	\begin{enumerate}[leftmargin=*,itemsep=3pt,topsep=2pt,parsep=0pt,partopsep=0pt,label=\rm{\roman*})] 
		\item \efx up to a factor of $0.618$, \label{part:thm_main_a}
		\item \gmms up to a factor of $0.553$ (thus, ditto for \mms), \label{part:thm_main_b}
		\item \efo, and 
		\item \pmms up to a factor of $0.667$.		
	\end{enumerate}
\end{mainthm}

We view parts \ref{part:thm_main_a} and \ref{part:thm_main_b} of breaking the $1/2$-approximation barrier for \efx and \gmms, as the  highlights of this work. These desirable properties are attained by Algorithm \ref{alg:efx} (Section \ref{sec:EFX}). We also suggest variations with improved guarantees for one notion at the expense of the others. The factors achieved by Algorithm \ref{alg:efx} and its variants, compared against the state of the art for each notion, are shown in Table \ref{tab:results}.

\begin{table}[ht]
	\centering
	\begin{tabular}{@{}lllll@{}}
		\toprule						  & \efx      & \efo      & \gmms     & \pmms  \\ \midrule
		{\small Best known (poly-time)}          & {$0.5$}   & {$1$}     & {$0.5$}   & {$0.781$} \\ 
		{\small Algorithm \ref{alg:efx}}         & {$0.618$} & {$1$}     & {$0.553$} & {$0.667$} \\ 
		{\small Variant in Thm.~\ref{thm:gmms2}} & {$0.6$}   & {$1$}     & {$0.571$} &  {$0.667$} \\ 
		{\small Variant in Thm.~\ref{thm:pmms2}} 
		& {$0.618$} & {$0.894$} & {$0.553$} & {$0.717$} \\ 
		\bottomrule
	\end{tabular}
	\caption{{\small Summary of our results and state of the art. Known results in the first row are due to \citet{PR18}, \citet{LMMS04}, \citet{BBMN18}, and \citet{Kurokawa17}, respectively. }}
	\label{tab:results}
\end{table}

At a technical level, our results are making use of two algorithms that are known to produce only \efo allocations. The first one is a simple draft algorithm and the second one is the envy-cycle-elimination algorithm of \citet{LMMS04}. Although these algorithms on their own do not possess \emph{any} good approximations with respect to \efx or \gmms, our main insight is that by carefully combining  parametric versions of these algorithms, we can obtain approximation guarantees for all the fairness criteria of interest here.   

In Section \ref{sec:m=n+2}, we return to the intriguing issue of existence. We show that $\gmms$---and thus $\pmms$ and $\efx$---allocations always exist, and can be found efficiently, when the number of goods does not exceed the number of agents by more than two. While this is a simple case, it is still non-trivial to tackle and has remained unresolved. Quite surprisingly, the idea of envy cycle elimination again comes to the rescue, after we carefully alter a small part of the instance. 
\medskip 


\noindent\textbf{Related work.} Envy-freeness was initially suggested by \citet{GS58}, and more formally by \citet{Foley67} and \citet{Varian74}. Regarding the relaxations of envy-freeness, \efo was  defined by \citet{Budish11}, but it was also implicit in the work of \citet{LMMS04}. Budish also defined the notion of maximin shares, based on concepts by \citet{Moulin90}. Later on, \citet{CaragiannisKMPS19} introduced the notions of \efx and \pmms, and even more recently, \citet{BBMN18} proposed to study \gmms allocations. Further variants and generalizations of the criteria we present here have also been considered, see e.g., \citet{Suk18}.

\efo allocations are known to be efficiently computable by the envy-cycle-elimination algorithm of \citet{LMMS04}. For all other notions, the focus has been on approximation algorithms since existence is either not guaranteed or is still an open problem. The most well studied notion is \mms with a series of positive results \citep{AMNS17,KurokawaPW18,BM17,GMT19}, 
and best known approximation of $3/4$ \citep{GHSSY18,GargTaki19}.
Exact and approximate \efx allocations with both additive and general valuations were studied by \citet{PR18}, achieving the currently best $1/2$-approximation. Recently, a polynomial time algorithm with the same guarantee has been obtained by \citet{CCLW19}. The same factor is also the best known for \gmms allocations in polynomial time by \citet{BBMN18} via a variant of envy cycle elimination. 
Finally, the currently best approximation of $0.781$ for \pmms is due to \citet{Kurokawa17}, using an approach similar to ours. Connections between the approximate versions of these criteria have been investigated by  \citet{AmanatidisBM18}; see \ref{app:connections} for a comparison of these implications with our results. 

Some of these fairness criteria have also been studied in combination with other objectives, such as Pareto optimality \citep{BMV18},
truthfulness \citep{AmanatidisBM16,ABCM17}
or maximizing the Nash welfare \citep{CaragiannisKMPS19,CaragiannisGH19,CKMS19}. 

In a very recent manuscript, \citet{ChaudhuryGM2020} show that \efx allocations always exist for instances with three agents and additive valuation functions.
While the existence of \efx allocations remains an open problem for more than three additive agents, it is possible to compute such allocations when it is allowed to discard some of the goods \cite{CaragiannisGH19,CKMS19} or when the valuation functions are special cases of additive functions \cite{AleksandrovWalsh2019,AmanatidisBRHV20}.
In \cite{CKMS19}, in parallel and independently of our work, 
\citet{CKMS19} also improve the $1/2$ factor for \gmms but not in polynomial time. The focus of their work is different and involves algorithms for \efx allocations by discarding a relatively small number of items. As an implication of their main results, they too follow closely the proof of Proposition 3.4 of \citep{AmanatidisBM18}, albeit from a different starting point, and obtain a pseudo-polynomial time $4/7$-approximation algorithm. This matches the guarantee of our Theorem \ref{thm:gmms2} for \gmms. However, to the best of our knowledge, this algorithm does not directly provide any good approximation for \efx. 


\section{Preliminaries}
\label{sec:prelims}
Let $N = \{1, 2, \ldots, n\}$ be a set of $n$ agents and $M$ be a set of $m$ indivisible items. 
Unless otherwise stated, we assume that each agent is associated with a monotone, {\it additive} valuation function, i.e., for $S\subseteq M$, 
$v_i(S) = \sum_{g\in S} v_i(\{g\})$. For simplicity, we write $v_{i}(g)$ instead of $v_i(\{g\})$, for $g\in M$. Monotonicity in this additive 
setting is equivalent to all items being \emph{goods}, i.e., $v_{i}(g)\geq 0$ for every $i\in N, g\in M$. For  the algorithms presented in this work, we assume that their input contains the valuation function of each involved agent, i.e., $v_i(g)$ is given to the algorithm for every agent $i$ and good $g$.

We consider the most standard setting in fair division, where we want to allocate all the goods to the agents (\emph{no free disposal}). An allocation of $M$ to the $n$ agents is therefore a partition,
$\mathcal{A} = (A_1,\ldots,A_n)$, where $A_i\cap A_j = \emptyset$ and $\cup_i A_i = M$.
By $\Pi_n(M)$ we denote  the set of all partitions of a set $M$ into $n$ bundles.

Although we allow for multiple goods to have the exact same value for a specific agent, we assume a deterministic tie-breaking rule for the goods (e.g., break ties lexicographically). This way we may abuse the notation and write $g = \argmax_{h\in M} v_{i}(h)$ instead of ``let $g$ be the lexicographically first element of $\argmax_{h\in M} v_i(h)$''.

\subsection{Fairness Concepts}
\label{subsec:concepts}
All the fairness notions we work with are relaxations of the classic notion of envy-freeness. 

\begin{definition}
	\label{def:EF}
	An allocation $\mathcal{A} = (A_1,\ldots,A_n)$ is envy-free (\ef), if for every $i, j\in N$, $v_i(A_i) \geq v_i(A_j)$.
\end{definition}

As envy-freeness is too strong to ask for, when we deal with indivisible goods, several relaxed fairness notions have been introduced so as to obtain meaningful positive results.
We start with two additive relaxations, 
and their approximate versions,
where an agent may envy another agent, but only by an amount dependent on the value of a single good in the other agent's bundle. 
\begin{definition}\label{def:EF1-EFX}
	An allocation $\mathcal{A} = (A_1,\ldots,A_n)$ is an
	\begin{enumerate}[leftmargin=*,itemsep=3pt,topsep=2pt,parsep=0pt,partopsep=0pt,label=\rm{\alph*})]
		\item $\alpha$-\efo allocation ($\alpha$-envy-free up to one good), if for every pair of agents $i, j\in N$, with $A_j\neq\emptyset$, there exists a good $g\in A_j$, such that
		$v_i(A_i) \geq \alpha\cdot v_i(A_j\setminus \{g\})$.
		\item  $\alpha$-\efx allocation ($\alpha$-envy-free up to any good), if for every pair $i, j\in N$, with $A_j\neq\emptyset$ and every good $g\in A_j$, it holds that $v_i(A_i) \geq \alpha\cdot v_i(A_j\setminus \{g\})$.\footnote{\ The original definition required  the condition to hold for all $g\in A_j$ with $v_{i}(g) >0$. This is often dropped in the literature, under the assumption that  all values are  positive \citep{PR18,CaragiannisGH19}. For our work \emph{neither} assumption is needed.} \label{def:EFX} 
	\end{enumerate}
\end{definition}
Of course, for $\alpha=1$ we obtain precisely the notions of envy-freeness up to one good (\efo) \citep{Budish11} and envy-freeness up to any good (\efx) \citep{CaragiannisKMPS19}. 
It is easy to see that \ef implies \efx, which in turn implies \efo. 

On a different direction, an interesting family of fairness criteria has been developed around the notion of \emph{maximin shares},  also proposed by \citet{Budish11}. The idea behind maximin shares is to capture the worst-case guarantees of generalizing the famous cut-and-choose protocol to multiple agents: Suppose agent $i$ is asked to partition the goods into $n$ bundles, while knowing that the other agents will choose a bundle before her. In the worst case, she will be left with her least valuable bundle. Assuming that agents are risk-averse, agent $i$ would choose a partition that maximizes the minimum value of a bundle. This gives rise to the following definition.

\begin{definition}
	\label{def:mmshare}
	Given $n$ agents, and a subset  $S\subseteq M$ of goods, the $n$-maximin share of agent $i$ with respect to $S$ is:
	\[ \bmu_i(n, S) = \displaystyle\max_{\mathcal{A}\in\Pi_n(S)} \min_{A_j\in \mathcal{A}} v_i(A_j)\,.\]
\end{definition}
From the definition, it directly follows that $n\cdot \bmu_i(n, S)\le v_i(S)$.
When $S=M$, this quantity is just called the \textit{maximin share} of agent $i$. 
We say that $\mathcal{T} \in \Pi_n(M)$ is an \emph{$n$-maximin share defining partition} for agent $i$, if $\min_{T_j\in \mathcal{T}} v_i(T_j) = \bmu_i(n, M)$.
When it is clear from context what $n$ and $M$ are, we simply write $\bmu_i$ instead of $\bmu_i(n, M)$.

The most popular fairness notion based on maximin shares, referred to as \emph{maximin share fairness}, asks for a partition that gives each agent her (approximate) maximin share.
\begin{definition}
	\label{def:MMS}
	An allocation 
	$\mathcal{A} = (A_1,\ldots,A_n) $ is called an $\alpha$-\mms ($\alpha$-maximin share) allocation if $v_i(A_i)\geq \alpha\cdot \bmu_i\,$, for every $i\in N$.
\end{definition} 

Variations of maximin share fairness have also been proposed. Here we focus on two notable examples. 
The first one, \emph{pairwise maximin share fairness}, is related but not directly comparable to \mms and was introduced by \citet{CaragiannisKMPS19}. The idea is to demand an \mms-type guarantee but for any \emph{pair} of agents. That is, we can think of an agent $i$ as considering the combined bundle of herself and another agent and requesting to receive at least her maximin share of this bundle if split into two subsets.

\begin{definition}
	\label{def:PMMS}
	An allocation $\mathcal{A} = (A_1,\ldots,A_n) $ is called an $\alpha$-\pmms ($\alpha$-pairwise maximin share) allocation if for every pair of agents $i, j\in N$, 
	$v_i(A_i)\geq \alpha\cdot\bmu_i(2, A_i\cup A_j)$.
\end{definition} 

Taking this one step further, we can demand an allocation to have an \mms-type guarantee for \emph{any subset} of agents. This is referred to as \emph{groupwise maximin share fairness}, introduced by \citet{BBMN18}.

\begin{definition}
	\label{def:GMMS}
	An allocation $\mathcal{A} = (A_1,\ldots,A_n)$ is called an $\alpha$-\gmms ($\alpha$-groupwise maximin share) allocation if for every subset  of agents $N'\subseteq N$ and any agent $i\in N'$, 
	$v_i(A_i)\geq \alpha\cdot\bmu_i(|N'|, \cup_{j\in N'} A_j )$.
\end{definition}

In Definitions \ref{def:MMS}, \ref{def:PMMS}, and \ref{def:GMMS}, when $\alpha=1$, we refer to the corresponding allocations as \mms, \pmms, and \gmms allocations respectively.
It is clear that the notion of \gmms is stronger than both \mms and \pmms. Further, it has been observed that \ef is stronger than \gmms  \citep{BBMN18} and, when all values are positive, \pmms is stronger than \efx \citep{CaragiannisKMPS19}. 
It should be noted, however, that the approximate versions of  these notions are related in non-straightforward ways \citep{AmanatidisBM18} (see also \ref{app:connections}).

An example illustrating  the different fairness criteria can be found in \ref{app:example}.

\subsection{Known EF1 Algorithms}
\label{subsec:ef1_algs}

Among the fairness notions defined above, \efo is the only one for which we know that it can always be achieved. Furthermore, two simple algorithms are already known for computing such allocations in polynomial time. We state below a parametric version of these algorithms so that they can run for a limited number of steps or on a strict subset of the goods, as we are going to use them later as subroutines.

In order to define the \emph{envy-cycle-elimination algorithm} (Algorithm \ref{alg:ece}) of \citet{LMMS04}, we first need to introduce the notion of an \emph{envy graph}. Suppose we have a partial allocation $\mathcal{P} = (P_1,\ldots,P_n)$, i.e., an allocation of a strict subset of $M$. We define the directed envy graph $G_{\mathcal{P}} = (N, E_{\mathcal{P}})$, where $(i,j) \in E_{\mathcal{P}}$ if and only if agent $i$ currently envies agent $j$, i.e., $v_i(P_i)<v_i(P_j)$. Algorithm \ref{alg:ece} builds an allocation one good at a time; in each step, an agent that no one envies receives the next available good. To ensure that such an agent always exists, the algorithm identifies  cycles that are created in the envy graph and eliminates them by appropriately reallocating some of the current bundles.

\begin{algorithm}[ht]
	\DontPrintSemicolon 
		Construct the envy graph $G_{\mathcal{P}}$\;
		\For{every $g \in M'$ in lexicographic order}{
			\While{there is no node of in-degree 0 in $G_{\mathcal{P}}$ }{
				Find a cycle $j_1 \to j_2 \to \ldots \to j_r \to j_1$ in $G_{\mathcal{P}}$ \;
				$B = P_{j_1}$\;
				\For{$k=1$ to $r-1$}{
					$P_{j_k}=P_{j_{k+1}}$ \tcc*{{\small shift the bundles}}
				}
				$P_{j_r}=B$\;
				Update $G_{\mathcal{P}}$\;
			}
			Let $i\in N$ be a node of in-degree $0$ \label{line:source} \; $P_{i} = P_{i}\cup \{g\}$\;
			Update $G_{\mathcal{P}}$ \;
		}
		\Return $\mathcal{P}$ \; 
	\caption{Envy-Cycle-Elimination$(N, \mathcal{P}, M')$\newline {\small where $N$: set of agents, $\mathcal{P}$: initial partial allocation, $M'$: set of unallocated goods}} \label{alg:ece}
\end{algorithm} 

Regarding tie-breaking in line \ref{line:source} of the algorithm, we assume that agent $i$ is the lexicographically first node of $G_{\mathcal{P}}$ with in-degree 0. Below we summarize the main known properties of Algorithm \ref{alg:ece} that we will utilize in our analysis.

\begin{theorem}[Follows by \citet{LMMS04}]\label{thm:ece}
	Let $\mathcal{P}$ be any \efo partial allocation and $M'= M\mysetminus\cup_{i=1}^{n}P_i$. Then, 
	\begin{enumerate}[leftmargin=*,itemsep=3pt,topsep=2pt,parsep=0pt,partopsep=0pt,label=\rm{\alph*})]
		\item at the end of each iteration of the \emph{for} loop, the resulting partial allocation  is \efo. Hence, the algorithm terminates with an \efo allocation in polynomial time. This holds even for agents with general monotone valuation functions.
		\item Fix an agent $i$, and let $A_i$ be the bundle assigned to $i$ at the end of some iteration of the \emph{for} loop. If $A_i'$ is  assigned to $i$ at the end of a future iteration, then $v_i(A_i') \geq v_i(A_i)$.\label{fact:value}
	\end{enumerate}
\end{theorem}

The first property of Theorem \ref{thm:ece} simply says that the \efo property is maintained during the execution of the algorithm, given an initial \efo allocation. The second property states that agents only get happier throughout the course of the algorithm, since they keep getting better and better bundles. 

For additive valuation functions there is, in fact, an even simpler greedy algorithm, referred to in the literature as the \emph{round-robin algorithm}, or the \emph{draft algorithm} (Algorithm \ref{alg:rr}) that also outputs \efo allocations, see e.g., \citep{Markakis17-survey}. 
Given a fixed ordering of the agents, they simply pick their favorite unallocated good one by one, according to that ordering, until there are no goods left.

\begin{algorithm}[ht]
	\DontPrintSemicolon 
		$k=1$ \;
		\While{$M'\neq \emptyset$ and $\tau>0$}{
			$g =  \argmax_{h\in M'} v_{\ell[k]}(h)$ \;
			$P_{\ell[k]} = P_{\ell[k]}\cup \{g\}$ \;
			$M' = M' \mysetminus \{g\}$\;
			$k = k+1 \mod n$ \;
			$\tau = \tau -1$ \; 
		}
		\Return $\left( \mathcal{P}, M'\right) $ \; 
	\caption{Round-Robin$(N, \mathcal{P}, M', \ell, \tau)$ \newline
		{\small where $N$: set of agents, $\mathcal{P}$: partial allocation, $M'$: set of unallocated goods, $\ell$: an ordering of $N$, $\tau$: number of steps}} \label{alg:rr}
\end{algorithm} 

\begin{theorem}\label{thm:rr}
	Let $\ell$ be any ordering of $N$ and $\mathcal{P}_\emptyset= (\emptyset,\ldots,\emptyset)$. Then Algorithm \ref{alg:rr} with input $(N, \mathcal{P}_\emptyset,\allowbreak M, \ell, |M|)$ produces an \efo allocation in polynomial time. 
\end{theorem}

\section{A Simple Universally Fair Algorithm}
\label{sec:EFX}
As mentioned above, Algorithm \ref{alg:efx} is built on Algorithms \ref{alg:ece} and \ref{alg:rr}. In particular, it first runs a simple preprocessing step (Algorithm \ref{alg:preprocess}) that determines an appropriate ordering $\ell$ of the set of agents $N$. Then, it suffices to run only two rounds of the round-robin algorithm, once with respect to $\ell$ and once with respect to the reverse of $\ell$ (the second run is also restricted to a subset of the agents), and finally run the envy-cycle-elimination algorithm on the remaining instance.
It should be noted here that the preprocessing step is mostly introduced to facilitate the presentation and the analysis of the algorithm. As it can be seen by its description, Algorithm \ref{alg:preprocess} could be combined with the first run of the round-robin algorithm. 
Indeed, the final assignments for the $h_i$s in Algorithm \ref{alg:preprocess} are exactly the goods that the agents receive in line \ref{line:1st_rr} of Algorithm \ref{alg:efx} (see also Lemma \ref{lem:inequalities} in the next section).

\begin{algorithm}[ht]
	\DontPrintSemicolon 
		$(\ell, n') = \text{Preprocessing}(N, M)$ \;
		Let $\mathcal{A} = (A_1,\ldots,A_n)$ with $A_i=\emptyset$ for each $i\in N$ \;
		$\left( \mathcal{A}, M'\right) =  \text{Round-Robin}(N, \mathcal{A}, M, \ell, n)$ \label{line:1st_rr}\;
		$\ell^R = (\ell[n],\ell[n-1],\ldots,\ell[1])$\;
		$\left( \mathcal{A}, M'\right) =  \text{Round-Robin}(N, \mathcal{A}, M', \ell^R, n-n')$ \label{line:2nd_rr}\;
		$\mathcal{A} =  \text{Envy-Cycle-Elimination}(N, \mathcal{A}, M')$ \label{line:ece}\;
		\Return $\mathcal{A}$ \; 
	\caption{Draft-and-Eliminate$(N, M)$} \label{alg:efx}
\end{algorithm}

As it is customary, we use $\phi$ to denote the \emph{golden ratio}. That is, $\phi$ is the positive solution to the quadratic equation $x^2-x-1 = 0$. Recall that $\phi = \frac{1+\sqrt{5}}{2}\approx 1.618$ and that $\phi - 1 = \phi^{-1} \approx 0.618$.

\begin{algorithm}[ht]
	\DontPrintSemicolon 
		$L = \emptyset$;\ \  
		$A = N$;\ \  
		$k=1$ \;
		\While{$A\neq \emptyset$}{
			Let $i$ be the lexicographically first agent of $A$ \;
			$h_i =  \argmax_{g\in M} v_{i}(g)$ \label{line:h_i_def}\;
			$t_i = m-|M|+1$ \tcc*{{\small $i$'s timestamp}} \label{line:timestamp}
			Let $R = (N\mysetminus (A \cup L ))\cup \{i\}$\; 
			$j =  \argmax_{t\in R} v_{i}(h_t)$ \;
			\uIf{$\phi\cdot v_i(h_i) < v_i(h_j)$ \label{line:if_preproc}}{
				$h_i = h_j$ \label{line:h_i_ass}\;
				$L = L\cup \{i\}$ \label{line:L}\;
				$\ell[k] = i$\;
				$k = k+1$\;
				$A = (A\mysetminus \{i\})\cup \{j\}$\;
			}\Else{
				$A = A\mysetminus \{i\}$\;
				$M = M\mysetminus \{h_i\}$ \; 
			}
		}
		\For{every $i \in N\mysetminus L$ in order of increasing timestamp $t_i$}{
			$\ell[k] = i$\;
			$k = k+1$\;	
		}
		\Return $(\ell, |L|)$ \; 
	\caption{Preprocessing$(N, M)$} \label{alg:preprocess}
\end{algorithm}

Before we move to the analysis of our algorithm, it is useful to build some more intuition on how things work. The preprocessing part essentially reorders $N$ so that the first few agents (namely, the first $|L|$ agents) are \emph{quite happy} with their pick in the first round of the round-robin subroutine. For the remaining agents, we make sure that they get a second good before we move to the envy-cycle-elimination algorithm. To do so in a ``balanced'' way, these agents pick goods in reverse order. The resulting partial allocation, where everyone receives one or two goods, turns out to have all the fairness properties we want to achieve at the end, e.g., it is $(\phi -1)$-\efx with respect to the currently allocated goods. Crucially, we show that starting from there and then applying the envy-cycle-elimination algorithm on the remaining instance 
maintains these properties.

Coming back to the preprocessing part, the intuition is to simulate a first round of Algorithm \ref{alg:rr} and correct any occurrences of extreme envy. In particular, if an agent envies someone that chose before her by a factor greater than $\phi$, then she is moved to a position of high priority in the ordering that is created. The agents moved to the first positions during this process (i.e., agents in $L$) are guaranteed a good of high value in line \ref{line:1st_rr} of Algorithm \ref{alg:efx}. To counterbalance their advantage, they are not allowed to pick a second good later in line \ref{line:2nd_rr}. 

To see that Algorithm \ref{alg:efx} runs in polynomial time, given the properties we have seen for Algorithms \ref{alg:ece} and \ref{alg:rr}, it suffices to  check that the preprocessing step can be efficiently implemented. Indeed, the \emph{if} branch of the while loop in Algorithm \ref{alg:preprocess} may be executed at most $n$ times, since agents are irrevocably added to $L$. Similarly, the \emph{else} branch may be executed at most $n$ times, as each time the set $A$ becomes smaller and its size never increases in the other parts of the algorithm. 

\section{Fairness Guarantees of Algorithm \ref{alg:efx}}
\label{sec:guarantees}
We begin our analysis with two useful lemmata about Algorithm \ref{alg:preprocess}. We stress that within Algorithm \ref{alg:preprocess}, every agent $i$ is associated with a distinct good $h_i$, although nothing is allocated at this step. The first lemma  establishes some useful inequalities regarding the goods associated with the agents. The second lemma states that this association actually coincides with the partial allocation produced in line \ref{line:1st_rr} of Algorithm \ref{alg:efx}. 

Recall that the set $L$, defined in Algorithm \ref{alg:preprocess}, contains the agents that get to pick first in line \ref{line:1st_rr} of Algorithm \ref{alg:efx} at the expense of not choosing a second good in line \ref{line:2nd_rr}. In terms of Algorithm \ref{alg:efx}, $L = \{\ell[1], \ell[2],\ldots,\ell[n']\}$). The partition of $N$ into $L$ and $N\mysetminus L$ is pivotal for distinguishing the different cases that are relevant in the analysis. 

\begin{lemma}\label{lem:inequalities}
	At the end of the execution of Algorithm \ref{alg:preprocess} with input $(N, M)$, each agent $i$ is associated with a single good $h_i$, so that
	\begin{enumerate}[leftmargin=*,itemsep=3pt,topsep=2pt,parsep=0pt,partopsep=0pt,label=\rm{\alph*})]
		\item $v_i(h_i)> \phi \cdot v_i(g)$, for any $i \in L$ and $g\in M\mysetminus \cup_{k=1}^{n}\{h_k\}$,\label{lem:ineq_1}
		\item $\phi \cdot v_i(h_i) \ge  v_i(h_j)$, for any $i, j \in N\mysetminus L$. \label{lem:ineq_2}
	\end{enumerate}
\end{lemma}

\begin{proof}
	a) Fix some $i\in L$ and consider the last iteration of the \emph{while} loop where $i$ was the lexicographically first agent of $A$, during which $i$ was added to $L$. We make the distinction between the initial and the final good associated with $i$ during this iteration by using $h_i^\text{old}$ and $h_i$ respectively to denote them. So, initially, $i$ was associated with $h_i^\text{old}$ which was her favorite good at the time, among the available ones. Since $i$ was eventually added to $L$, we know that the condition in line \ref{line:if_preproc} was true. That is, $i$'s favorite good among the ones associated to an agent not in $L$, say $h_j$, was more than $\phi$ times more valuable than $h_i^\text{old}$. By the choice of $h_i^\text{old}$, we have $v_i(h_j) > \phi\cdot v_i(h_i^\text{old}) \ge  \phi\cdot v_i(g)$ for any good $g$ that was not associated to an agent at the time. Note that the set of unassociated goods during the execution of the algorithm only shrinks, and thus the last inequality also holds for any good $g$ that was not associated to any agent till the end. Finally, recall that $h_i= h_j$, as imposed by line \ref{line:h_i_ass},  to conclude that $v_i(h_i)  >  \phi\cdot v_i(g)$ for any $g\in M\mysetminus \cup_{k=1}^{n}\{h_k\}$.\smallskip
	
	b) Fix some $i, j \in N\mysetminus L$. Note that both $i$ and $j$ may be considered multiple times during Algorithm \ref{alg:preprocess}, as they may be removed and then added back to the set $A$ several times. We consider two cases, based on the \emph{last} time that each agent was considered (i.e., the last time each of them was the lexicographically first agent in $A$). If the last time that $i$ was considered by Algorithm \ref{alg:preprocess} happened before the last time that $j$ was considered (i.e., $t_i<t_j$ at the end), the desired inequality is straightforward, as agent $i$ is associated with her favorite good among the available goods, $h_i$, before agent $j$, i.e., $v_i(h_i) \ge  v_i(h_j)$. On the other hand, if the last time $i$ was considered by Algorithm \ref{alg:preprocess} takes place after the last time that $j$ was considered (i.e., $t_i>t_j$ at the end), suppose that $\phi \cdot v_i(h_i) <  v_i(h_j)$. Then, during the last iteration that $i$ was considered, line \ref{line:if_preproc} would be true and $i$ would be (irrevocably) added to $L$ in line \ref{line:L}, which contradicts the choice of $i$. 
\end{proof}

In order for the above lemma to be of any use, we need a connection between the $h_i$s and the partial allocations that are produced in the first part of Algorithm \ref{alg:efx} (lines \ref{line:1st_rr}-\ref{line:2nd_rr}). At a first glance, the issue  is that the order in which the goods are assigned in Preprocessing$(N, M)$ is somewhat different than the order in which the goods are allocated in Round-Robin$(N, \mathcal{A}, M, \ell, n)$. Next we establish this connection.

\begin{lemma}\label{lem:1st_rr}
	The partial allocation produced in line \ref{line:1st_rr} of Algorithm \ref{alg:efx} is $\mathcal{A} = (\{h_{1}\},\{h_{2}\}, \allowbreak \ldots, \allowbreak \{h_{n}\})$, where the $h_i$s are as in Lemma \ref{lem:inequalities}.
\end{lemma}

\begin{proof}
	We start with the easy observation that the ordering $\ell$ that is used in the Round-Robin subroutine in line \ref{line:1st_rr} of Algorithm \ref{alg:efx} \emph{is not} the same with the order that goods get assigned to agents during the preprocessing step, even when one takes into account that moving agents into $L$ is similar to changing their order. 
	
	First, using induction, we are going to show that agents in $L$ get assigned to them (in Algorithm \ref{alg:preprocess}) the same goods they would if they were to choose first (in Algorithm \ref{alg:rr}) according to $\ell$. For agent $\ell[1]$ this is straightforward, since she gets her favorite good in both cases. So, assume for our inductive step that agents $\ell[1],\ldots,\ell[k-1]$ did receive goods $h_{\ell[1]},\ldots,h_{\ell[k-1]}$ respectively in Algorithm \ref{alg:rr}. We argue that agent $\ell[k]$'s favorite available good from $M_k = \allowbreak M\mysetminus \cup_{i=1}^{k-1}\{h_{\ell[i]}\}$ is $h_{\ell[k]}$.  
	Indeed, in the execution of Algorithm \ref{alg:preprocess}, consider the last iteration of the \emph{while} loop where $\ell[k]$ was the lexicographically first agent of $A$. At the time, $L$ was $\{\ell[1],\ldots,\ell[k-1]\}$ and in lines \ref{line:h_i_def}-\ref{line:h_i_ass} agent $\ell[k]$ gets assigned $h_{\ell[k]}$ which is her favorite good among the ones assigned outside $L$ exactly because this was much better than her favorite unassigned good. That is, (the final choice for) $h_{\ell[k]}$ is $\ell[k]$'s favorite good from $M_k$. This concludes the inductive step.

	Given the above, it is not hard to argue about agents in $N\mysetminus L$. 
	First, observe that, on termination, agents in $N\mysetminus L$ all have distinct timestamps assigned in line \ref{line:timestamp}. Although this is not true for all agents, an agent in $N\mysetminus L$ only gets an existing timestamp, 
	if this belongs to one or more agents already in $L$. For the remainder of the proof, by \emph{timestamp} of an agent $i$, we mean the \emph{final value of $t_i$}.
	Now, fix some $i\in N\mysetminus L$. 
	In Algorithm \ref{alg:preprocess} agent $i$ gets assigned her favorite good available if we exclude the goods assigned to \emph{some} agents in $L$, and to all the agents in $N\mysetminus L$ with timestamp less than $t_i$ (\emph{first scenario}). In Algorithm \ref{alg:rr}, agent $i$ receives her favorite good available if we exclude the goods allocated to \emph{all} the agents in $L$, and to all the agents in $N\mysetminus L$ with timestamp less than $t_i$ (\emph{second scenario}). However, in the second scenario, the extra agents from $L$ that pick before $i$ choose goods that $i$ does not find attractive enough. Indeed, those goods are available in the first scenario (where they are the only extra available options compared to the second scenario), yet $i$ does not prefer them to $h_i$. Therefore, in the second scenario $i$ also picks $h_i$ and $A_i = \{h_i\}$ after line \ref{line:1st_rr} of Algorithm \ref{alg:efx}.
\end{proof}

Given Lemmata \ref{lem:inequalities} and \ref{lem:1st_rr}, we are going to consistently use the $h_i$ notation for the goods allocated in line \ref{line:1st_rr} of Algorithm \ref{alg:efx}, throughout the remaining of this section. Further, for the agents who receive a second good in line \ref{line:2nd_rr} of Algorithm \ref{alg:efx} we use $h_i'$ to denote that second good of agent $i$.

As a warm-up we first obtain that Algorithm \ref{alg:efx} maintains the fairness guarantee of its components, i.e., \efo fairness.

\begin{proposition}\label{prop:efo}
	Algorithm \ref{alg:efx} returns an \efo allocation.
\end{proposition}

\begin{proof}
	By Theorem \ref{thm:ece}, it suffices to show that the partial allocation $\mathcal{A} = (A_1, \ldots, A_n)$ produced in line \ref{line:2nd_rr} of Algorithm \ref{alg:efx} is \efo.
	Fix two distinct agents $i, j \in N$. If $j\in L$, then $A_j = \{h_j\}$. Clearly, $v_i(A_i) \ge v_i(A_j\mysetminus \{h_j\}) = 0$. On the other hand, if $j\in N\mysetminus L$, then $A_j = \{h_j, h_j'\}$. Since agent $i$ picked $h_i$ when $h_j'$ was still available, $v_i(h_i) \ge v_i(h_j')$. So, we have $v_i(A_i) \ge v_i(h_i) \ge v_i(A_j\mysetminus \{h_j\})$.
\end{proof}

\subsection{Envy-Freeness up to Any Good}
\label{sec:efx}
Proving whether
\efx allocations always exist or not seems very challenging \citep{PR18,CaragiannisGH19,ChaudhuryGM2020}. Even improving on the $1/2$ approximation factor of \citet{PR18} has been one of the most intriguing recent open problems in fair division. In this sense, we view the following as one of the highlights of this work. 

\begin{theorem}\label{thm:efx}
	The allocation $\mathcal{A} = (A_1, \ldots, A_n)$ returned by Algorithm \ref{alg:efx} is a $(\phi - 1)$-\efx allocation.
\end{theorem}

\begin{proof}
	Consider the allocation $\mathcal{A} = (A_1, \ldots, A_n)$ returned by the algorithm, and fix two distinct agents $i, j \in N$. If $|A_j| = 1$, then clearly, $v_i(A_i) \ge \max_{g\in A_j} v_i(A_j\mysetminus \{g\}) = 0$. 
	So, assume that $|A_j| \ge 2$ and let $h$ be the last good added to $A_j$ (either in line \ref{line:2nd_rr} by reverse round-robin or in line \ref{line:ece} by envy-cycle-elimination). 
	Of course, at the time this happened, $A_j$ may belonged to an agent $j'$ other than $j$.  Finally, let $A_i^{\text{old}}$, $A_{j'}^{\text{old}}$ be the bundles of $i$ and $j'$, respectively, right before $h$ was allocated (i.e., $h$ was added to $A_{j'}^{\text{old}}$). 
	Note that $A_i^{\text{old}}$ may not necessarily be a subset of $A_i$ due to the possible swaps imposed by Algorithm \ref{alg:ece}, but part \ref{fact:value} of Theorem \ref{thm:ece}  
	implies that $v_i(A_i)\ge v_i(A_i^{\text{old}})$. We consider four cases, depending on whether $i\in L$ and on the type of step during which $h$ was added to $A_{j'}^{\text{old}}$.
	\smallskip
	
	\noindent\ul{Case 1 ($i\in L$ and $h$ added in line \mbox{\ref{line:2nd_rr}}).} We have $A_i^\text{old} = \allowbreak \{h_i\}$, as well as  $j' \in N\mysetminus L$ and $A_j = \allowbreak \{h_{j'}, h_{j'}'\}$. This immediately implies that $v_i(A_i^{\text{old}}) \ge \allowbreak  \max\{v_i(h_{j'}), v_i(h_{j'}')\}$ and, thus, $v_i(A_i) \ge \max_{g\in A_j} v_i(A_j\mysetminus \{g\})$.\footnote{\ Here we achieve the \efx objective exactly. Instead, in a similar argument as in Case 2, we could have used that $v_i(h_i) > v_i(h_{j'})$ and  $v_i(h_i) > \phi\cdot v_i(h'_{j'})$  to get $v_i(A_i) \ge (\phi-1) v_i(A_j)$.}
	\medskip
	
	\noindent\ul{Case 2 ($i\in L$ and $h$ added in line \mbox{\ref{line:ece}}).} By the way that envy-cycle-elimination chooses who to give the next good to, (line \ref{line:source} of Algorithm \ref{alg:ece}), we know that right before $h$ was added, no one envied $j'$. In particular, $v_i(A_i^{\text{old}}) \ge v_i(A_{j'}^{\text{old}})$. We further have $v_i(A_i^{\text{old}}) \ge v_i(h_i) > \phi\cdot v_i(h)$, where the last inequality directly follows from part \ref{lem:ineq_1} of Lemma \ref{lem:inequalities}. Putting everything together,
	\[v_i(A_j) = v_i(A_{j'}^{\text{old}}) + v_i(h) \le (1+\phi^{-1}) v_i(A_i^{\text{old}}) \le \phi\cdot v_i(A_i),\]
	or, equivalently, $v_i(A_i) \ge \phi^{-1} \cdot v_i(A_j) = (\phi-1) v_i(A_j)$.
	\medskip
	
	\noindent\ul{Case 3 ($i\notin L$ and $h$ added in line \mbox{\ref{line:2nd_rr}}).} We have $i, j' \in N\mysetminus L$ and $A_j = \allowbreak \{h_{j'}, h_{j'}'\}$. If $\ell[i]<\ell[j']$, then we proceed in a way similar to Case 1. Indeed, 
	\begin{IEEEeqnarray*}{rCl}
		v_i(A_i) & \ge & v_i(A_i^{\text{old}}) \ge v_i(h_i) \ge \max\{v_i(h_{j'}), v_i(h_{j'}')\} 
		 =  \max_{g\in A_j} v_i(A_j\mysetminus \{g\}) .
	\end{IEEEeqnarray*}
	So, assume that $\ell[i]>\ell[j']$. This, in particular, means that $v_i(h_i') \ge v_i(h_{j'}')$. We have
	\begin{IEEEeqnarray*}{rCl}
		v_i(A_i) & \ge & v_i(A_i^{\text{old}}) \ge v_i(h_i) + v_i(h_i') \ge \frac{1}{\phi} v_i(h_{j'}) + v_i(h_{j'}')\\
		& \ge & \frac{1}{\phi} (v_i(h_{j'}) + v_i(h_{j'}'))  =  (\phi-1) v_i(A_j),
	\end{IEEEeqnarray*}
	where the third inequality directly follows from part \ref{lem:ineq_2} of Lemma \ref{lem:inequalities}.
	\medskip

	\noindent\ul{Case 4 ($i\notin L$ and $h$ added in line \mbox{\ref{line:ece}}).} Arguing like in Case 2, we have $v_i(A_i^{\text{old}}) \ge v_i(A_{j'}^{\text{old}})$. Moreover, by the way round-robin works, we know that $v_i(h_i)\ge v_i(h_i') \ge v_i(h)$. In particular, $v_i(h) \le \frac{1}{2} v_i(\{h_i, h_i'\}) \le \frac{1}{2} v_i(A_i^{\text{old}})$. 
	Putting things together, we have
	\[v_i(A_j) = v_i(A_{j'}^{\text{old}}) + v_i(h) \le \Big( 1+\frac{1}{2}\Big)  v_i(A_i^{\text{old}}) \le \phi\cdot v_i(A_i).\]
	Equivalently, $v_i(A_i) \ge (\phi-1) v_i(A_j)$.
\end{proof}

It is not hard to see that our analysis is tight, i.e.,  there are instances (even with $n=2$ and $m=4$) for which the resulting allocation is not $(\phi - 1 +\varepsilon)$-\efx for any $\varepsilon >0$.

\subsection{Groupwise Maximin Share Fairness}
\label{sec:gmms}
A result of \citet{AmanatidisBM18} (Proposition 3.4) implies that every \emph{exact} \efx allocation is also a $4/7$-\gmms allocation. Of course, the allocation produced by Algorithm \ref{alg:efx} is not exact \efx and, in general, 
an arbitrary $(\phi - 1)$-\efx allocation need not even be a $0.404$-\gmms allocation (see  \ref{app:connections}). 
For the particular allocation returned by Algorithm \ref{alg:efx}, however, we can show that the \gmms guarantee is significantly better. Parts of our proof closely follow the proof of the aforementioned  proposition of \citep{AmanatidisBM18}.
 
Before moving on to the technical details, it is worth mentioning that our analysis does not yield a stronger guarantee for the weaker notion of \mms. This is partly because it relies on guarantees for approximate \ef or \efx (following from the proof of Theorem \ref{thm:efx}) which in turn hold for any subset of agents. Moreover, the algorithm itself is designed with \efx as the primary goal and, thus, it is lacking specific subroutines that are known to yield \mms guarantees, like bipartite matchings and bag-filling \citep{KurokawaPW18,AMNS17,BM17,GargMT19,GHSSY18,GargTaki19}.

We are going to need the following simple lemma that allows to remove appropriately chosen subsets of goods, while reducing the number of agents, so that  the maximin share of a specific agent does not decrease. In particular, the lemma implies that for any good $g$, $\bmu_i(n-1, M\mysetminus \{g\}) \geq \bmu_i(n, M)$.

\begin{lemma}[\citet{AmanatidisBM18}]\label{lem:monotonicity}
	Suppose $\mathcal{T} \in \Pi_n(M)$ is an $n$-maximin share defining partition for agent $i$. 
	Then, for any set of goods $S$, such that there exists some $j$ with $S \subseteq T_j$, it holds that $\bmu_i(n-1, M\mysetminus S) \geq \bmu_i(n, M)$.
\end{lemma} 

\begin{theorem}\label{thm:gmms}
	The allocation $\mathcal{A} = (A_1, \ldots, A_n)$ returned by Algorithm \ref{alg:efx} is a $\frac{2}{\phi +2}$-\gmms allocation.
\end{theorem}

\begin{proof}
	Suppose that $\mathcal{A}$ is not a $\frac{2}{\phi +2}$-\gmms allocation, i.e., 
	there exists a subset of agents $Q\subseteq N$ with $|Q|=q$, and some agent $j\in Q$, so that $v_{j}(A_j)< \frac{2}{\phi +2} \bmu_j(q, R)$, where $R = \cup_{k\in Q} A_k$. That is, with respect to $Q$ and $R$, the restriction of  $\mathcal{A}$ to $Q$ is not a $\frac{2}{\phi +2}$-\mms allocation.
	To facilitate the presentation, and without loss of generality, we may assume that $Q= [q]$ and that agent $1$ is such a ``dissatisfied'' agent. We write $\bmu_1$ instead of $\bmu_1(q, R)$.

	We may remove any agent in $Q$, \emph{other} than agent 1, that receives exactly one good, 
	and still end up with a suballocation that is not a $\frac{2}{\phi +2}$-\gmms allocation. 
	Indeed, if $|A_i|=1$ for some $i\in Q\mysetminus \{1\}$, then $(A_1, \ldots, A_{i-1}, A_{i+1}, \ldots, A_q)$ 
	is an allocation of $R\mysetminus A_i$ to $Q\mysetminus\{i\}$ and, 
	by Lemma \ref{lem:monotonicity}, $\bmu_1' = \allowbreak \bmu_1(q-1, R\mysetminus A_i) \geq \bmu_1$. 
	Thus, $v_1(A_1)< \frac{2}{\phi +2} \bmu_1'$.
	Therefore, again without loss of generality, we may assume that $|A_i|\ge 2$ for all $i\in Q\mysetminus \{1\}$ 
	in the initial allocation $\mathcal{A}$. 
	At this point, we make the distinction on whether $1\in L$ or not. 
	\smallskip
	
	\noindent\ul{Case 1 ($1\in L$).}
	As we see from (the footnote of) Case 1 and from Case 2 of the proof of Theorem \ref{thm:efx}, we always have $v_1(A_1) \ge (\phi-1) v_1(A_i)$ (or equivalently $v_1(A_i) \le \phi  v_1(A_1)$) for all $i\in Q\mysetminus \{1\}$.
	Recall that, by the definition of maxi\-min share, $\bmu_1 \le \frac{1}{q}v_1(R)$. Thus 
	\[	q \bmu_1  \le v_1(R) = \sum_{k\in Q} v_1(A_k) \le q \phi v_1(A_1).\]
	That is, we get $v_1(A_1) \ge (\phi -1) \bmu_1\ge \frac{2}{\phi +2} \bmu_1$, which contradicts the choices of $\mathcal{A}$ and $A_1$.	
	\medskip

	\noindent\ul{Case 2 ($1\notin L$).}
	Consider some $i\in Q\mysetminus \{1\}$ and let $h$ be the last good added to $A_i$. Following the notation introduced in the proof of Theorem \ref{thm:efx}, this bundle belonged to some agent $i'$ and $A_{i'}^{\text{old}}$ denotes the bundle allocated to $i'$ right before $h$ was added. According to Cases 3 and 4 in the proof of Theorem \ref{thm:efx}, if $h$ was added in line \ref{line:2nd_rr} of Algorithm \ref{alg:efx} and $\ell[1]\ge\ell[i']$ or if it was added in line \ref{line:ece}, then $v_1(A_i) \le \phi  v_1(A_1)$.
	We still need to deal with the subcase where $h$ was added in line \ref{line:2nd_rr} but $\ell[1] < \ell[i']$. We call such an $A_i$ \emph{dubious}. For dubious bundles, by their definition, we directly have $|A_i| = 2$ and $v_1(A_1)\ge v_1(h_1) \ge \max_{g\in A_i}v_1(g)$.  
	If a bundle $A_i$ is not dubious, or if it is dubious but we have $v_1(A_i) \le \frac{3}{2}  v_1(A_1) < \phi  v_1(A_1)$, we say that $A_i$ is \emph{convenient}.  
	A (dubious) bundle is \emph{inconvenient} if it is not convenient. A good is inconvenient if it belongs to an inconvenient bundle. Let $B$ be the set of all inconvenient goods.

	Now we are going to show that $v_1(A_1) \ge \frac{2}{\phi +2} \bmu_1(q', R')$ for a reduced instance that we get by possibly removing some inconvenient goods. We do so in a way that ensures that $\bmu_1(q', R')\ge \bmu_1$, thus contradicting  the choices of $\mathcal{A}$ and $A_1$. 
	We consider a $q$-maximin share defining partition $\mathcal{T}$ for agent $1$ with respect to $R$, i.e., $\min_{T_i\in \mathcal{T}} v_1(T_i) = \bmu_1$ and $\cup_{k\in Q} T_k = R$.\footnote{\ Note that while goods in $B$ may have no apparent significance for the reduced instances and the allocations we talk about from this point on, we keep referring to them as \emph{inconvenient}.}
	If there is a bundle of $\mathcal{T}$ containing two goods of $B$, $g_1$, $g_2$, then we remove those two goods and reduce the 
	number of agents by one. By Lemma \ref{lem:monotonicity}, we have that $\bmu_1(q-1, R\mysetminus \{g_1, g_2\}) \ge \bmu_1$.
	We repeat as many times as necessary to get a reduced instance with $q'\le q$ agents and a set of 
	goods $R'\subseteq R$ for which there is a $q'$-maximin share defining partition $\mathcal{T'}$ for agent $1$, such that no bundle contains more than \emph{one} good from $B$.  By repeatedly using Lemma \ref{lem:monotonicity}, we get $\bmu_1(q', R') \geq \bmu_1$. 
	
	Let $x$ be the number of goods from $B$ in the reduced instance. 
	Clearly, $x$ cannot be greater than $q'$, or some bundle of $\mathcal{T'}$ would contain at least $2$ inconvenient goods. Further, if $|B| = y$, i.e., the number of inconvenient goods in the original instance, then we know that the number of convenient bundles in the restriction of $\mathcal{A}$ on $Q$ was $q-\frac{y}{2}$, and that the number of agents was reduced $\frac{y-x}{2}$ times, i.e., $q'= q-\frac{y-x}{2}$. That is, we can express the number of convenient bundles in the original instance in terms of $q'$ and $x$ only, as $q'-\frac{x}{2}$.

	In order to upper bound $v_1(R')$, notice that $R'$ contains all the goods of all the convenient bundles plus $x$ inconvenient goods. 
	Recall that any good of a dubious bundle has value at most $v_1(A_1)$, and that if $A_i$ is convenient then $v_1(A_i) \le \phi  v_1(A_1)$.
	So, we have
	\begin{IEEEeqnarray*}{rCl}
		v_1(R') & \le & x v_1(A_1) + \left( q'-\frac{x}{2} -1\right) \phi v_1(A_1) + v_1(A_1) \nonumber\\
		& = & \left( \phi q'+ (1 - \phi / 2)x - (\phi -1)\right)  v_1(A_1) \nonumber\\
		& \le & \left( \phi + (1 - \phi / 2)\right) q' v_1(A_1) =  \frac{\phi +2}{2} q' v_1(A_1).
	\end{IEEEeqnarray*} 
	Combining this inequality with $\bmu_1 \le \bmu_1(q', R')$ (by the construction of the reduced instance) and $\bmu_1(q', R') \le \frac{1}{q'}v_1(R')$ (by the definition of maxi\-min share), we get
	$v_1(A_1) \ge \frac{2}{\phi +2} \bmu_1$, which contradicts the choices of $\mathcal{A}$ and $A_1$.	
\end{proof}

A natural question is why we do not achieve the factor of $4/7$ of the original proposition of \citep{AmanatidisBM18} instead. 
A close inspection of the original proof reveals that we need a slightly stronger upper bound for the value of the convenient bundles, i.e., a factor of $3/2$ rather than $\phi$ that we actually have here. There is no easy way to fix this \emph{in general} without other things breaking down badly in the analysis of Algorithm \ref{alg:efx}. The crucial observation, however, is that we only need the distinction of convenient  and inconvenient bundles for agents in $N\mysetminus L$. By fine-tuning line \ref{line:if_preproc} of Algorithm \ref{alg:preprocess}, we are able to improve the inequalities about the convenient bundles just for agents in $N\mysetminus L$ and obtain a $4/7$-\gmms allocation, at the expense of some loss with respect to \efx.

\begin{theorem}\label{thm:gmms2}
	Suppose we modified Algorithm \ref{alg:efx} by changing $\phi$ in  line \ref{line:if_preproc} of Algorithm \ref{alg:preprocess} to $3/2$. Then the resulting allocation is a $4/7$-\gmms allocation. It is also a $3/5$-\efx, a $2/3$-\pmms, and an \efo allocation.
\end{theorem}

\begin{proof}
	The proof for \efo goes through as is.
	We only need to highlight the differences in comparison to the proofs of Theorems \ref{thm:efx}, \ref{thm:pmms} and \ref{thm:gmms}. We begin with \efx and  Theorem \ref{thm:efx}. When $|A_j|=1$ it is again straightforward that $v_i(A_i) \ge \max_{g\in A_j} v_i(A_j\mysetminus \{g\})$. In Cases 1 and 2 (using the footnote of Case 1), by substituting $2/3$ for $\phi$, we get $v_i(A_i) \ge \frac{3}{5} v_i(A_j)$. The first part of Case 3 (where $\ell[i]<\ell[j']$) is the same, giving $v_i(A_i) \ge \max_{g\in A_j} v_i(A_j\mysetminus \{g\})$, while in the second part (where $\ell[i]>\ell[j']$) we get $v_i(A_i) \ge \frac{2}{3} v_i(A_j)$ by using $3/2$ instead of $\phi$.
	Case 4 is exactly the same, giving $v_i(A_i) \ge  \frac{2}{3} v_i(A_j)$. So, for any pair $i, j$ of agents, we either have $v_i(A_i) \ge \max_{g\in A_j} v_i(A_j\mysetminus \{g\})$ or $v_i(A_i) \ge  \frac{3}{5} v_i(A_j)$, implying that  $v_i(A_i) \ge \frac{3}{5} \max_{g\in A_j} v_i(A_j\mysetminus \{g\})$.
	
	Moving to \pmms and Theorem \ref{thm:pmms}, by following the same proof and using the $3/5$ factor from above instead of $\phi$ in the last couple of lines, we still get the same factor $2/3$.
	
	Finally, for \gmms and Theorem \ref{thm:gmms}, we may follow the same arguments with $4/7$ rather than $\phi/(\phi +2)$. For Case 1, using the inequality $v_1(A_1) \ge \frac{3}{5} v_1(A_i)$ from the \efx case above, we have $v_1(A_1) \ge \frac{3}{5} \bmu_1\ge \frac{4}{7} \bmu_1$. For Case 2, the above analysis for \efx implies that when $1\notin L$ and $A_i$ is not dubious we have the even stronger guarantee: $v_1(A_1) \ge  \frac{2}{3} v_1(A_i)$. This directly implies that $v_1(A_1) \ge  \frac{2}{3} v_1(A_i)$ for any \emph{convenient} bundle $A_i$. Coming to the final argument for upper bounding $v_1(R')$, we have
	\begin{IEEEeqnarray*}{rCl}
		v_1(R') & \le & x v_1(A_1) + \left( q'-\frac{x}{2} -1\right) \frac{3}{2} v_1(A_1) + v_1(A_1) \nonumber\\
		& = & \left( \frac{3q'}{2}+\frac{x}{4} -\frac{1}{2} \right)  v_1(A_1) 
		\le  \frac{7}{4} q' v_1(A_1). 
	\end{IEEEeqnarray*} 
	Like before, we combine this inequality with $\bmu_1 \le \bmu_1(q', R')$ and $\bmu_1(q', R') \le \frac{1}{q'}v_1(R')$  to get that
	$v_1(A_1) \ge \frac{4}{7} \bmu_1$.	
\end{proof}

\begin{remark}
As we already mentioned, our approach in both theorems is heavily based on the \ef or \efx guarantees we get in the various different cases and depends less on the details of the algorithms themselves. For that reason the analyses are probably not tight. 
On the other hand, for any $\varepsilon>0$, there are instances for which Algorithm \ref{alg:efx} does not return a $(2/3 + \varepsilon)$-\gmms allocation, independently of how we may tune the constant in line \ref{line:if_preproc} of Algorithm \ref{alg:preprocess}.
While we omit the description of such instances since they do not establish tightness of the algorithm, we suspect that the actual approximation ratio of the modified Algorithm \ref{alg:efx} in Theorem \ref{thm:gmms2} for \gmms is $2/3$.  
\end{remark}

\subsection{Pairwise Maximin Share Fairness}
\label{sec:pmms}
Any result for \gmms directly translates to a result for \pmms with the exact same guarantee. Note, however, that the proof of Theorem \ref{thm:gmms} suggests that the bad event with respect to \gmms is having many inconvenient bundles. When we only deal with two agents at a time, it is not hard to see that inconvenient bundles are not an issue. In fact, their existence would not be able to force the approximation ratio for \pmms below $2/3$, even if the goods where divisible. 
Indeed, following the cases in the proof of Theorem \ref{thm:efx} it is relatively easy to show that this is exactly the guarantee achieved by Algorithm \ref{alg:efx}.

\begin{theorem}\label{thm:pmms}
	The allocation $\mathcal{A} = (A_1, \ldots, A_n)$ returned by Algorithm \ref{alg:efx} is a  $2/3$-\pmms allocation. 
\end{theorem}

\begin{proof}
	Fix two distinct agents $i, j \in N$. 
	To show the desired approximation ratio we are going to use the analysis in the proof of Theorem \ref{thm:efx} and the simple inequality $2\bmu_i(2, S)\le v_i(S)$ that directly follows from Definition \ref{def:MMS} for any subset $S$ of goods.
	To simplify the notation, we use $A_{ij}$ as a shorthand for $A_i \cup A_j$.  
	
	If $|A_j| = 1$, then by Lemma \ref{lem:monotonicity}, $\bmu_i(2, A_{ij})\le \bmu_i(1, A_i) = v_i(A_i)$. 
	So, assume that $|A_j| \ge 2$ and let $h$ be the last good added to $A_j$.  We adopt the notation of the proof of Theorem \ref{thm:efx} here as well. That is, we assume that right before $h$ was added, $A_j$ belonged to some agent $j'$.  
	We first examine the case where $i\in N\mysetminus L$, $h$ was added in line \ref{line:2nd_rr} of Algorithm \ref{alg:efx} by reverse round-robin, and $\ell[i]<\ell[j']$. This case is the \emph{bottleneck} for achieving a better approximation factor, as the rest of the analysis reveals.	
	Here we know that $A_j = \{h_{j'}, h_{j'}'\}$ and that $\max\{v_i(h_{j'}), v_i(h_{j'}')\}\le v_i(h_i)$. Therefore,
	\[2\bmu_i(2, A_{ij})\le v_i(A_{ij}) \le v_i(A_i) + 2 v_i(h_i) \le 3 v_i(A_i) . \]
	It immediately follows that $v_i(A_i) \ge \frac{2}{3}\bmu_i(2, A_{ij}).$
	
	For all the other cases, by examining the proof of Theorem \ref{thm:efx}, we can see that $v_i(A_j) \le \phi v_i(A_i)$.
	Therefore,
	\[2\bmu_i(2, A_{ij})\le v_i(A_{ij}) \le v_i(A_i) + \phi v_i(A_i) \le 3 v_i(A_i) . \]
	Again, it follows that $v_i(A_i) \ge \frac{2}{3}\bmu_i(2, A_{ij})$.
\end{proof}

While the above factor is tight for our algorithm, it is easy to see that if we exclude the bottleneck case in the proof of Theorem \ref{thm:pmms}, then the approximation ratio goes up to $\frac{2}{\phi +1} \approx 0.764$. Hence, we could try to improve this single problematic case where both $i$ and $j'$ receive two goods from the round-robin subroutine but $j'$ has lower priority. Note that the bundles of $i$ and $j'$  start off well, i.e., right after line \ref{line:2nd_rr} of Algorithm \ref{alg:efx} we have $v_i(A_i) \ge \bmu_i(2, A_{ij'})$. The issue is that during the envy-cycle-elimination phase, $A_i$ might be updated to a bundle that still has value almost $v_i(h_i)$ but can be combined with $h_{j'}$ and $h_{j'}'$ to produce two sets of value roughly $\frac{3}{2} v_i(h_i)$. To remedy that, we can modify slightly the envy graph. The high level idea---due to \citet{Kurokawa17}---is that an agent from $N\mysetminus L$ should only exchange her \emph{initial} bundle of two goods for something significantly better. 

Suppose we start with a partial allocation $\mathcal{P} = (P_1,\ldots,P_n)$ produced in line \ref{line:2nd_rr} of Algorithm \ref{alg:efx}. For $\alpha > 1$, the $\alpha$-modified envy graph  $G_{\mathcal{P}}^{\alpha}$ is defined like the envy graph $G_{\mathcal{P}}$ but we drop any edge $(i,j)$ where:  $i\in N\mysetminus L$, \emph{and} $i$ still has her original bundle,  \emph{and}  $\alpha\cdot v_i(P_i) > v_i(P_j)$. That is, agents in $N\mysetminus L$ are represented in the envy graph as having an artificially amplified value (by a factor of $\alpha$) specifically for their original bundles.

The following theorem indicates how far we can push the appro\-ximation factor for \pmms,  at the expense of \efo, while preserving the original guarantees with respect to \efx and \gmms.

\begin{theorem}\label{thm:pmms2}
	Suppose we modified Algorithm \ref{alg:efx} 
	by using the $(\phi-\frac{1}{2})$-adjusted envy graph in Algorithm \ref{alg:ece}. Then the resulting allocation is a  $\frac{4\phi - 2}{2\phi +3}$-\pmms and a $\frac{2}{2\phi - 1}$-\efo allocation. Moreover, the guarantees of Theorems \ref{thm:efx} and \ref{thm:gmms} are not affected. 
\end{theorem}

\begin{proof}
	First, it is easy to see the \efo guarantee. Before running the modified envy-cycle-elimination algorithm, the allocation is \efo. Then, whenever a new good gets allocated, it is given to an agent that no one envies by more than a factor of $\frac{2\phi-1}{2}$. That is, for any agents $i, j \in N$, if $h$ is the last good given to $j$, then $v_i(A_i) \ge \frac{2}{2\phi-1} v_i(A_j\mysetminus \{h\})$.
	
	For the other notions, we only need to highlight the differences from the proofs of Theorems \ref{thm:efx}, \ref{thm:pmms} and \ref{thm:gmms}. We begin with \efx and  Theorem \ref{thm:efx}. It is easy to see that the \emph{only} step that differs is Case 4.  
	It is not anymore the case that $v_i(A_i^{\text{old}}) \ge v_i(A_{j'}^{\text{old}})$. Instead, $(\phi-\frac{1}{2}) v_i(A_i^{\text{old}}) \ge v_i(A_{j'}^{\text{old}})$. 
	Also, like in the original proof,  we have $v_i(h) \le  \frac{1}{2} v_i(A_i^{\text{old}})$. 
	Putting everything together, we have
	\begin{IEEEeqnarray*}{rCl}
		v_i(A_j) & = & v_i(A_{j'}^{\text{old}}) + v_i(h) \le \Big( \phi-\frac{1}{2}+\frac{1}{2}\Big)  v_i(A_i^{\text{old}}) 
		 =   \phi\cdot v_i(A_i).
	\end{IEEEeqnarray*}
	We conclude that we have a $(\phi-1)$-\efx allocation.
	
	Given the guarantee for \efx, the proof of Theorem \ref{thm:gmms} for \gmms goes through as is.
	
	Finally, for \pmms and Theorem \ref{thm:pmms} we may follow the same general proof except for some details.
	We fix two agents $i, j \in N$.
	Like in the proof of Theorem \ref{thm:pmms}  we use $A_{ij}$ to denote $A_i \cup A_j$.  
	If $|A_j| = 1$, the original argument holds, so we assume that $|A_j| \ge 2$. Let $h$ be the last good added to $A_j$ and assume that right before $h$ was added, $A_j$ belonged to some agent $j'$. 
	
	Again, we first go over  the bottleneck  case where $i\in N\mysetminus L$, $h$ was added in line \ref{line:2nd_rr} of Algorithm \ref{alg:efx} by reverse round-robin, and $\ell[i]<\ell[j']$. Then $A_j = \{h_{j'}, h_{j'}'\}$ and  $\max\{v_i(h_{j'}), v_i(h_{j'}')\}\le v_i(h_i)$. A vital  distinction now is whether $A_i$ is $i$'s original bundle from line \ref{line:2nd_rr} of Algorithm \ref{alg:efx}. Suppose this is the case, i.e.,  $A_i = \{h_i, h'_i\}$ and $A_{ij} = \{h_i, h'_i, h_j, h'_j\}$. Then, $h_i = \argmax_{g\in A_{ij}}v_i(g)$ and it is not hard to see that 
	$v_i(A_i) \ge \bmu_i(2, A_{ij})$ (see also Lemma \ref{lem:4_mms}\ \ref{part:4_mms_a} within the proof of Theorem \ref{thm:m=n+2}). Next, suppose $A_i$ is not $i$'s original bundle. Then, for the modified  envy-cycle-elimination algorithm to give $i$ another bundle, it must be the case that $v_i(A_i) \ge \frac{2\phi-1}{2} v_i(\{h_i, h'_i\})$. Therefore,
	\begin{IEEEeqnarray*}{rCl}
		2\bmu_i(2, A_{ij}) & \le  & v_i(A_{ij}) \le v_i(A_i) + 2 v_i(h_i) 
		 \le  v_i(A_i) + 2 v_i(\{h_i, h'_i\})  \\
		& \le & v_i(A_i) + \frac{4}{2\phi-1} v_i(A_i)
		 =   \frac{2\phi +3}{2\phi-1} v_i(A_i).
	\end{IEEEeqnarray*}
	It follows that $v_i(A_i) \ge \frac{4\phi-2}{2\phi +3}\bmu_i(2, A_{ij}).$
	
	For all the other cases,  we have that $v_i(A_j) \le \phi v_i(A_i)$.
	Therefore, $2\bmu_i(2, A_{ij})\le v_i(A_{ij}) \le v_i(A_i) + \phi v_i(A_i)$ and 
	it follows that $v_i(A_i) \ge \frac{2}{\phi +1}\bmu_i(2, A_{ij}) \ge \frac{4\phi-2}{2\phi +3}\bmu_i(2, A_{ij}).$
	
	We conclude that the allocation is $\frac{4\phi-2}{2\phi +3}$-\pmms.
\end{proof}


\section{GMMS, PMMS, and EFX with a Few Goods}
\label{sec:m=n+2}
In this section we focus on the exact versions of the fairness notions under consideration. In particular, we  
show that \gmms allocations always exist when $m\le n+2$.
This implies that \pmms and \efx allocations also exist for this case by the discussion in Section \ref{sec:prelims}.\footnote{\ Actually, the existence of \efx allocations is directly implied by the existence of \pmms allocations \emph{only} when all values are positive. However, our result is more general.} 

As it is indicated in the proof of Theorem \ref{thm:m=n+2}, the interesting case is when $m = n+2$ and is tackled by Algorithm \ref{alg:m=n+2}. When $m\le n$ the problem is trivial, and the $m= n+1$ case is rather straightforward as well. Adding one extra good, however, makes things significantly more complex. To point out how challenging these simple restricted cases can be, we note that for the much better studied notion of \mms fairness it is still open whether exact \mms allocations exist when $m=n+5$ \citep{KurokawaPW18}.

Quite surprisingly, the envy-cycle-elimination algorithm again comes to rescue for the case when $m=n+2$. We first run the round-robin algorithm to allocate  $n-1$ goods to the first $n-1$ agents. After this, we have 3 goods remaining. Allocating these goods to the last agent may destroy the properties we are after, so we need to be careful on how to handle these three goods. 
Instead, we (pretend to) pack them into two  boxes; the big box (i.e., the virtual good $p$) ``contains'' two goods and the small box (i.e., the virtual good $q$) ``contains'' one. We tell each agent separately that the big box contains her favorite two out of the three items and give the big box to the last agent. Then we proceed using the envy-cycle-elimination algorithm. At the end, the owner of the big box gets her two favorite goods, while the owner of the small box gets the remaining good. 

\begin{algorithm}[h!]
	\DontPrintSemicolon 
		Let $\ell = (1,2,\ldots,n)$ and $\mathcal{A} = (\emptyset,\ldots,\emptyset)$ \;
		$\left( \mathcal{A}, M'\right) =  \text{Round-Robin}(N, \mathcal{A}, M, \ell, n-1)$ \label{line:rr} \;
		Create two virtual goods $p$ and $q$, such that for all $i\in N$: \newline $v_i(q) = \min_{g\in M'} v_i(g)$ and $v_i(p) = v_i(M') - v_i(q)$\;
		Allocate $p$ to agent $n$\;
		\vspace{1pt}$\mathcal{A} = \text{Envy-Cycle-Elimination}(N, \mathcal{A}, \{q\})$ \label{line:ece_pq}\;
		Give to the final owner of $p$ her two favorite goods from $M'$ and to the final owner of $q$ the remaining good \label{line:alloc_pq}\;
		\Return $\mathcal{A}$ \; 
	\caption{Draft-Pack-and-Eliminate$(N, M)$} \label{alg:m=n+2}
\end{algorithm}

\begin{theorem}\label{thm:m=n+2}
	For instances with $m\le n+2$, a \gmms allocation always exists and can be efficiently computed.
\end{theorem}

\begin{proof}
When $m\le n$, we arbitrarily allocate one good to each agent, till there are no goods left, to produce $\mathcal{A} = (A_1, \ldots, A_n)$. Fix a subset $S$ of agents, and let $i\in N$. The combined bundle of all agents in $S$ either contains strictly less than $|S|$ goods or exactly $|S|$ goods. In the first case, we trivially have $\bmu_i(|S|, \cup_{j\in S} A_j) = 0$, whereas in the second case, we have $\bmu_i(|S|, \cup_{j\in S} A_j) = \min_{g\in \cup_{j\in S} A_j} v_i(g)$. In both cases, we have that $v_i(A_i) \geq \bmu_i(|S|, \cup_{j\in S} A_j)$, and $\mathcal{A}$ is a \gmms allocation.

For $m>n$, we will use the following simple observations. 
	\begin{lemma}\label{lem:max_mms}
		Let $Q\subseteq N$ 
		and $T\subseteq M$ such that $|T| = 2 |Q| - 1$. Then, for any $i\in Q$, we have
		$\max_{g\in T}v_i(g) \ge \bmu_i(|Q|, T)$.
	\end{lemma}
	
\begin{proof}[Proof of Lemma \ref{lem:max_mms}]
By the pigeonhole principle, in any possible partition of $T$ in $|Q|$ parts, at least one bundle will have at most $1$ good. So, in any $|Q|$-maximin share partition of $T$ for an agent $i\in Q$, her worst bundle's worth is upper bounded by $\max_{g\in T}v_i(g)$.
\end{proof}
	
The next one is established within the proof of Theorem 5.1 of \citep{AmanatidisBM16}.
	
	\begin{lemma}
		\label{lem:4_mms}
		Let $i \in N$ 
		and $T= \{g_1, g_2, g_3, g_4\} \subseteq M$ such that $v_i(g_1)\ge v_i(g_2) \ge v_i(g_3) \ge v_i(g_4)$. Then
		\begin{enumerate}[leftmargin=*,itemsep=3pt,topsep=2pt,parsep=0pt,partopsep=0pt,label=\rm{\roman*)}]
			\item $v_i(\{g_1, g_4\}) \ge \bmu_i(2, T)$ \label{part:4_mms_a}, and
			\item $\max\big\{v_i(g_1), v_i(\{g_2, g_3\})\big\} \ge \bmu_i(2, T)$. \label{part:4_mms_b}
		\end{enumerate}	
	\end{lemma}

When $m =  n +1$, 
let $\mathcal{A} = (\emptyset,\ldots,\emptyset)$, i.e., $A_i=\emptyset$ for all $i\in N$, and $\ell = (1,2,\ldots,n)$, i.e., the standard lexicographic order. We can first run Round-Robin$(N, \mathcal{A}, M, \ell, n)$ and then give the remaining good to agent $n$ to get the final allocation $(\{h_1\},\ldots,\{h_{n-1}\},\{h_n, h_n'\})$. 
Consider a subset of agents $S$ and an agent $i\in S$. From $S$, we can eliminate anyone, other than $i$, that owns at most $1$ good, without reducing $i$'s maximin share. That is, $v_i(A_i) = \bmu_i(1, A_i) \ge \bmu_i(|S|, \cup_{j\in S} A_j)$, where the last inequality follows after $|S| - 1$ applications of Lemma \ref{lem:monotonicity}. Hence, if no agent in $S\mysetminus \{i\}$ got $2$ goods, then we would be done. The problem then reduces to checking the case where $S$ consists of agents $i$ and $n$. In particular, we need to show that $v_i(h_i) \ge \bmu_i(2, \{h_i, h_n, h_n'\})$
for all $i <n$. Indeed, given that $v_i(h_i)\ge \max\{v_i(h_n), v_i(h_n')\}$, 
this directly follows from Lemma \ref{lem:max_mms}.
Hence, $\mathcal{A}$ is a \gmms allocation.
	
When $m =  n +2$, we use Algorithm \ref{alg:m=n+2} to compute the allocation $\mathcal{A} = (A_1, \ldots, A_n)$. 
We consider two cases, depending on whether  $\mathcal{A}$ contains a bundle with $3$ goods or not. 
\smallskip

	\noindent\ul{Case 1 (One agent receives $3$ goods).}
	Let  $\hat{\jmath}$ be the agent for whom $|A_{\hat{\jmath}}| = 3$. 
	By arguing like before about repeatedly eliminating agents that received exactly $1$ good via Lemma \ref{lem:monotonicity}, it is easy to see that the problem of whether $v_i(A_i)\geq \bmu(|S|, \cup_{j\in S} A_j )$, for any $S\subseteq N$ and any $i\in S$, is reduced to whether $v_i(A_i) \ge \bmu_i(2, A_i \cup A_{\hat{\jmath}})$ for any $i\in N\mysetminus\{ \hat{\jmath}\}$.
	
	The only way that $\hat{\jmath}$ ended up with $3$ goods is if she received both $p$ and $q$ and $A_{\hat{\jmath}} = M'$. When $q$ was allocated, some agent $\hat{\jmath}\,'$ (possibly other than $\hat{\jmath}\,$) had $p$. Given that no one envied $\hat{\jmath}\,'$ at the time and that the envy-cycle-elimination never decreases the value of an agent's bundle (Theorem \ref{thm:ece}\ \ref{fact:value}), for any $i\in N\mysetminus \{\hat{\jmath}\}$ we have
	\[v_i(A_i)\ge v_i(p) = v_i(A_{\hat{\jmath}}) - \min_{g\in A_{\hat{\jmath}}} v_i(g) = \max_{g\in A_{\hat{\jmath}}} v_i(A_{\hat{\jmath}}\mysetminus \{g\}) .\]
	Using part \ref{part:4_mms_b} of Lemma \ref{lem:4_mms}, this implies that $v_i(A_i) \ge \bmu_i(2, A_i \cup M')$. So, in this case $\mathcal{A}$ is a \gmms allocation.

	\medskip

	\noindent\ul{Case 2 (Two agents receive $2$ goods each).}
	Let  $\hat{\jmath}_p$ be the agent who ended up with $p$ and $\hat{\jmath}_q$ be the agent who ended up with $q$ (in addition to her 
	other good
	$d$) after line \ref{line:ece_pq} of Algorithm \ref{alg:m=n+2}. Clearly,  
	these are the only agents who receive $2$ goods.
	Further, let $A_{\hat{\jmath}_p} = \{a, b\}$ and  $A_{\hat{\jmath}_2} = \{c, d\}$, i.e., $M' = \{a, b, c\}$ and $a, b$ are $\hat{\jmath}_p$'s two favorite goods from $M'$. 
	
	Again, by repeatedly using Lemma \ref{lem:monotonicity}, the problem reduces to a small number of subcases involving at most one agent that received $1$ good. Specifically, to show that $\mathcal{A}$ is a \gmms allocation, it suffices to show
	\begin{enumerate}[leftmargin=*,itemsep=3pt,topsep=2pt,parsep=0pt,partopsep=0pt,label=(\roman*)]
		\item $v_i(A_i) \ge \bmu_i(2, A_i \cup A_{\hat{\jmath}_p})$ for any $i\in N\mysetminus \{\hat{\jmath}_p, \hat{\jmath}_q\}$ \label{part:m=n+2_a}
		\item $v_i(A_i) \ge \bmu_i(2, A_i \cup A_{\hat{\jmath}_q})$ for any $i\in N\mysetminus \{\hat{\jmath}_p, \hat{\jmath}_q\}$ \label{part:m=n+2_b}
		\item $v_i(A_i) \ge \bmu_i(3, A_i \cup A_{\hat{\jmath}_q}\cup A_{\hat{\jmath}_q})$ for any $i\in N\mysetminus \{\hat{\jmath}_p, \hat{\jmath}_q\}$\label{part:m=n+2_c}
		\item $v_{\hat{\jmath}_q}(A_{\hat{\jmath}_q}) \ge \bmu_{\hat{\jmath}_q}(2, A_{\hat{\jmath}_p} \cup A_{\hat{\jmath}_q})$ \label{part:m=n+2_d}
		\item $v_{\hat{\jmath}_p}(A_{\hat{\jmath}_p}) \ge \bmu_{\hat{\jmath}_p}(2, A_{\hat{\jmath}_p} \cup A_{\hat{\jmath}_q})$ \label{part:m=n+2_e}		
	\end{enumerate}
	We first deal with subcases \ref{part:m=n+2_a}, \ref{part:m=n+2_b} and \ref{part:m=n+2_c}. By how round-robin works, we know that $i$ preferred her initial good to any good in $M'$. Moreover, by how envy-cycle-elimination works, we know that $i$ did not envy the bundle $\{d\}$ right before $q$ was added to it and that her current good is no worse than any good she had earlier. Thus, $i$'s current good is her favorite good in $A_i \cup M'$. Then   \ref{part:m=n+2_a}, \ref{part:m=n+2_b} and \ref{part:m=n+2_c} all follow by Lemma \ref{lem:max_mms}.
	
	We next move to subcase \ref{part:m=n+2_d}. Like above, because of round-robin, we know that $\hat{\jmath}_q$ preferred her initial good to any good in $M'$ and, because of envy-cycle-elimination, her good $d$ is no worse than her initial good. Thus, $d$ is $\hat{\jmath}_q$'s favorite good in $A_{\hat{\jmath}_p} \cup A_{\hat{\jmath}_q}$. Even if $c$ is her least favorite good, \ref{part:m=n+2_d} directly follows from part \ref{part:4_mms_a} of Lemma \ref{lem:4_mms}.
	
	Finally, we consider  subcase \ref{part:m=n+2_e}. If $A_{\hat{\jmath}_p}$ contains $\hat{\jmath}_p$'s favorite good in $A_{\hat{\jmath}_p} \cup A_{\hat{\jmath}_q}$, then \ref{part:m=n+2_e}  follows from part \ref{part:4_mms_a} of Lemma \ref{lem:4_mms}. So, suppose that $\hat{\jmath}_p$'s favorite good is either $c$ or $d$. By line \ref{line:alloc_pq} of Algorithm \ref{alg:m=n+2}, we know that $\min\{v_{\hat{\jmath}_p}(a), v_{\hat{\jmath}_p}(b)\} \ge v_{\hat{\jmath}_p}(c)$. Thus, it must be that $d$ is $\hat{\jmath}_p$'s favorite good and that $a, b$ are her second and third favorite goods. Moreover, $\hat{\jmath}_p$ did not envy the bundle $\{d\}$ right before $q$ was added to it and $\{a,b\}$ is no worse than any bundle she had earlier. That is,  $v_{\hat{\jmath}_p}(\{a,b\})\ge v_{\hat{\jmath}_p}(d)$. Then \ref{part:m=n+2_e} follows from part \ref{part:4_mms_b} of Lemma \ref{lem:4_mms}.
	
	So, $\mathcal{A}$ is a \gmms allocation in this case as well.
\end{proof}

\begin{corollary}\label{cor:m=n+2}
	When $m\le n+2$, we can efficiently find \pmms and \efx allocations.
\end{corollary}

\begin{proof}
	For \pmms this is trivial. Given that \citet{CaragiannisKMPS19} show that each \pmms allocation is an \efx allocation when all values are positive, we directly get the existence of \efx allocations for this special case. However, going through the proof of Theorem \ref{thm:m=n+2}, it is not hard to see that the result holds under the  stronger Definition \ref{def:EF1-EFX}\ \ref{def:EFX}.  
	
	Fix two distinct agents $i, j \in N$. As usual, if $|A_j| = 1$, then $v_i(A_i) \ge \max_{g\in A_j} v_i(A_j\mysetminus \{g\}) = 0$, so assume that $|A_j| \ge 2$.  
	
	First consider the case where $|A_j| = 2$ and let $|A_j|=\{x,y\}$. If $j$ is the agent who ended up with $p$, then we directly have  $v_i(A_i) \ge v_i(h_i) \ge \max\{v_i(x), v_i(y)\}$. So suppose $j$ is the agent who ended up with $q = x$ in addition to her other good $y$. First, notice that $v_i(A_i) \ge v_i(h_i) \ge v_i(x)$. Because Algorithm \ref{alg:ece} added $x$ to $\{y\}$, we also get that $v_i(A_i) \ge v_i(y)$. We conclude that the allocation is \efx.

	We then consider the case where $|A_j| = 3$. That is, $j$ received both $p$ and $q$ and $A_j = M'$. But since Algorithm \ref{alg:ece} added $q$ to $p$, we have
	\[v_i(A_i) \ge v_i(p) = v_i(M') - v_i(q) = v_i(A_j) - \min_{g\in A_j} v_i(g).\]
	Again, the allocation is \efx.	
\end{proof}

\section*{Acknowledgments}
This work has been partly supported by the COST Action CA16228 ``European Network for Game Theory''.
G.~Amanatidis was partially supported by the NWO Gravitation project NETWORKS (No.~024.002.003), the ERC Advanced Grant  AMDROMA (No.~788893), the MIUR PRIN project ALGADIMAR and the NWO Veni project No.~VI.Veni.192.153.




\appendix

\section{An Example Illustrating the  Different Notions}
\label{app:example}
Suppose we have the following instance with three agents and five goods.
It can be verified that this instance does not admit an \ef allocation.

\begin{table}[ht]
	\centering
	\begin{tabular}{@{}cccccc@{}} 
		\toprule		& $a$ & $b$ & $c$ & $d$ & $e$ \\ \midrule
		Agent $1$ & $10$ & $6$ & $7$ & $5$ & $3$  \\
		Agent $2$ & $6$ & $8$ & $12$ & $7$ & $5$ \\
		Agent $3$ & $10$ & $11$ & $3$ & $2$ & $7$ \\
		\bottomrule
	\end{tabular}
	\label{table:Example2}
\end{table}

First consider  the allocation ${\mathcal A}= (A_1, A_2, A_3) = (\{a, d\}, \{c\}, \{b, e\})$. This is not an envy-free allocation, since $v_2(A_2) < v_2(A_1)$. However, we claim that it satisfies all the relaxed fairness criteria. 
\begin{itemize}[leftmargin=*,itemsep=3pt,topsep=2pt,parsep=0pt,partopsep=0pt]
	\item To see that $\mathcal A$ is \efx for agent 1, we have
	$v_1(A_1) = 15$ and, on the other hand, $v_1(A_2 \setminus \{ c \}) = 0$, $v_1(A_3 \setminus \{ b \}) = 3$ and $v_1(A_3 \setminus \{ e \}) = 6$.
	Thus, the \efx inequalities for agent $1$ are satisfied. Similarly, we can check that for agents $2$ and $3$ the corresponding inequalities are satisfied as well.
	\item To verify that $\mathcal{A}$ is \mms, \pmms, and \gmms, we first observe that $\bmu_1(3, M) = 10$, $\bmu_2(3, M) = 12$ and $\bmu_3(3, M) = 10$. Hence, we have that $v_i(A_i) \geq \bmu_i(3, M)$ for $i=1, 2, 3$, which establishes \mms. We can also examine the three different pairs of agents and see that the pairwise maximin shares are attained. For instance, $v_1(A_1) \geq \bmu_1(A_1\cup A_2) = 10$, and $v_1(A_1) \geq \bmu_1(A_1\cup A_3) = 11$. Hence, $\mathcal{A}$ is a \pmms allocation. As there is no other subset of agents to examine, this means $\mathcal{A}$ is also \gmms.
\end{itemize}

To demonstrate the approximate versions of the fairness criteria, consider now the allocation ${\mathcal A'} = (A_1', A_2', A_3') = (\{b\}, \{c, e\}, \{a, d\})$. 
\begin{itemize}[leftmargin=*,itemsep=3pt,topsep=2pt,parsep=0pt,partopsep=0pt]
	\item $\mathcal A'$ is neither \ef nor \efx, but it is easy to check that it is \efo. We also claim that it is a 0.6-\efx allocation. To see this, note that the approximation is due to agent 1, since agents 2 and 3 do not experience any envy. Observe that $v_1(A_1') = 0.6\cdot v_1(A_3' \setminus \{ d \}) \geq v_1(A_3' \setminus \{ a \})$, and that $v_1(A_1') \geq 0.6\cdot v_1(A_2' \setminus \{ e \}) \geq 0.6\cdot v_1(A_2' \setminus \{ c \})$. Hence, indeed, ${\mathcal A'}$ is a 0.6-\efx allocation.
	\item It is also easy to verify that ${\mathcal A'}$ is a 0.6-\mms allocation, and the same approximation holds for \pmms and \gmms. Again it suffices to see the performance of agent 1, and we have $v_i(A_1') \geq 0.6 \bmu_1(3, M)$, implying the approximation for \mms. It also holds that $v_1(A_1') = 0.6\cdot \bmu_1(A_1'\cup A_3') \geq 0.6\cdot \bmu_1(A_1'\cup A_2')$, which determines the approximation for \pmms and \gmms.
\end{itemize}


\section{On the Connections Between the Approximate Versions of Fairness Guarantees}
\label{app:connections}
Here we deal with  the question of whether the approximation guarantee established for one of the considered fairness notions directly implies the approximation guarantees for the other notions.  
For example, given that we obtain a $(\phi-1)$-\efx approximation allocation, does this immediately enforce any of the fairness guarantees that we show for \efo, \pmms and \gmms? 
If this were true, our task would become simpler. Some of our proofs would not be necessary and we would not have to analyze separately the approximation guarantees for all the fairness notions.

Such questions have been studied recently by \citet{AmanatidisBM18} who gave an almost complete picture regarding the relations among the notions of \efo, \efx, \mms, and \pmms. 
 
We exhibit here that none of our improved approximation factors can be derived by a black-box use of the guarantees provided in \citep{AmanatidisBM18}. 
Hence, the separate analyses of Algorithm \ref{alg:efx} for each notion are necessary for obtaining our results.

To begin with, suppose we have an \efo allocation. 
The results in Section 3 of \citet{AmanatidisBM18} state that an arbitrary \efo allocation cannot yield a constant factor approximation for \efx and \mms---hence neither for \gmms. It is also shown that every \efo allocation is $1/2$-\pmms, and this is tight, but this guarantee  is weaker than what we obtain here for \pmms.  

Moving on to approximate \efx allocations, Proposition 3.7 in \citep{AmanatidisBM18} states that an arbitrary $\alpha$-\efx allocation is also a ${2\alpha}/{(2+\alpha)}$-\pmms allocation, and this is tight. This means that our $(\phi-1)$-\efx approximation in Theorem \ref{thm:efx} immediately implies only a ${2(\phi -1)}/{(1+\phi)}\approx 0.472$-\pmms approximation, which is worse than the guarantee in Theorem \ref{thm:pmms}. 
Further, Proposition 3.5 in \citep{AmanatidisBM18} implies (via \mms) that an arbitrary $\alpha$-\efx allocation is not necessarily a $\max\{\frac{\alpha}{1+\alpha}, \frac{8\alpha}{11+2\alpha}\}$-\gmms allocation. For our case, this means that we cannot obtain a better than $0.404$-\gmms approximation directly from Theorem \ref{thm:efx}.

Suppose now that we obtain an $\alpha$-\mms allocation. Proposition 4.5 and Corollary 4.9 in \citep{AmanatidisBM18} show that this cannot yield a constant factor approximation for either \efx or \pmms. The same then is true for \gmms. 
Similarly, Propositions 4.4 and 4.8 in \citep{AmanatidisBM18} show that an $\alpha$-\pmms allocation does not necessarily yield a constant approximation for \efx or \mms, and thus neither for \gmms.
Hence, our Theorems \ref{thm:pmms} and \ref{thm:pmms2} cannot provide the guarantees we have in Theorems \ref{thm:efx}, \ref{thm:gmms} and \ref{thm:gmms2} for \efx and \gmms. 

The only implications left to examine is when we start with an arbitrary $\alpha$-\gmms allocation. Since \gmms allocations were not considered by \citet{AmanatidisBM18}, we show here the following result. 
\begin{proposition}\label{prop:a_gmms}
	Let $\alpha\in (0,1)$.	For $n\ge 2$, an $\alpha$-\gmms  allocation is not necessarily an $(\frac{\alpha}{2-\alpha}+\varepsilon)$-\efo or an $(\alpha+\varepsilon)$-\pmms, or a $\beta$-\efx allocation for any  $\varepsilon>0, \beta \in (0,1)$.
\end{proposition}

\begin{proof}
	Consider the following 
	instance with $n$ agents and $4k +n$ goods for $k\ge 1$. Let $\beta \in (0,1)$. We focus on agent $1$ and we have 
	\[ 
	v_1(g_j)= \left\{
	\begin{array}{cl}
	\frac{\alpha}{2k} & 1 \leq j\leq 2k\,, \vspace{3pt}\\
	\frac{2-\alpha}{2k} &2k+1 \leq j\leq 4k\,,  \vspace{3pt}\\
	0.01 & 4k+1\le j\le 4k+2\,,  \vspace{3pt}\\
	2/\beta &  j = 4k + 3\,, \vspace{3pt}\\
	1 & 4k+3 \le j \le 4k +n \,.
	\end{array} 
	\right. 
	\]
	Let $\mathcal{A} = (A_1, \ldots, A_n)= \allowbreak (\{g_1, \allowbreak \ldots, \allowbreak g_{2k}\}, \allowbreak \{g_{2k+1},  \allowbreak \ldots, \allowbreak g_{4k+1}\}, \allowbreak \{g_{4k+2},g_{4k+3}\}, \allowbreak \{g_{4k+4}\}, \allowbreak \ldots, \allowbreak \{g_{4k+n}\})$ and assume that agents $2$ through $n$ are not envious, i.e., they value their corresponding bundles much higher than everything else. It is only a matter of simple calculations to see that according to agent 1, this is an $\alpha$-\gmms allocation which is, however, only an $\frac{\alpha}{2-\alpha}$-\efo, an $\alpha$-\gmms and a $\frac{\beta}{2}$-\efx allocation. 
\end{proof}

Plugging $\alpha = \frac{2}{\phi +2}$ in Proposition \ref{prop:a_gmms}, we get that an arbitrary $0.553$-\gmms  allocation is not necessarily a $0.383$-\efo or a $0.554$-\pmms allocation, and has no guarantee whatsoever with respect to \efx. Hence, our Theorem \ref{thm:gmms} (or even Theorem \ref{thm:gmms2} for that matter) does not imply the guarantees we obtain for \efx and \pmms in Theorems \ref{thm:efx} and \ref{thm:pmms} respectively.


\bibliography{efx-gmms-references}

\end{document}